\documentclass[a4paper,10pt]{amsart}

\usepackage[utf8x]{inputenc}
\usepackage[american]{babel}

\usepackage{hyperref}
\usepackage{amssymb}
\usepackage{amsmath,amsfonts, mathrsfs}
\usepackage{xspace}
\usepackage{todonotes}
\usepackage{boxedminipage,verbatim,float,caption}
\usepackage{subcaption}
\usepackage{hyperref}
\usepackage{pstricks,pst-tree,pst-node} 
\DeclareMathOperator{\cost}{cost}
\DeclareMathOperator{\True}{True}
\DeclareMathOperator{\False}{False}
\newcommand{\itemb}{\item[$\bullet$]}
\newcommand{\dsum}[2]{\displaystyle{\sum_{#1}^{#2}} }
\newcommand*{\dcup}[2]{\displaystyle{\underset{#1}{\overset{#2}{\bigcup}}}}
\renewcommand{\bar}[1]{\bar{#1}}

\renewcommand{\leq}{\leqslant}
\renewcommand{\geq}{\geqslant}

\def\ie{\emph{i.e.}\xspace}

\def\cR{\mathcal{R}}

\renewcommand{\bar}[1]{\overline{#1}}

\newcommand{\pbDef}[3]{%
	\noindent
	\begin{center}
		\begin{boxedminipage}{ \columnwidth}
			\textbf{#1}\\[5pt]
			\textbf{Input:}  #2\\
			\textbf{Output:}  #3
		\end{boxedminipage}
	\end{center}
}

\theoremstyle{theorem}
\newtheorem{theorem}{Theorem}[section]
\newtheorem{lemma}[theorem]{Lemma}

\newtheorem{corollary}[theorem]{Corollary}

\theoremstyle{definition}
\newtheorem{definition}[theorem]{Definition}
\newtheorem{claim}[theorem]{Claim}

\begin{document}
	
	\title{On Minimum Connecting Transition Sets in Graphs}

	\author{Thomas Bellitto}
	\address{Université de Bordeaux, LABRI, CNRS, France}
	\email{thomas.bellitto@u-bordeaux.fr}
	 \author{Benjamin Bergougnoux}
	 \address{Université Clermont Auvergne, LIMOS, CNRS, France }
	 \email{benjamin.bergougnoux@gmail.com} 
	\thanks{This work is supported by French Agency for Research under the GraphEN project (ANR-15-CE-0009).}

	\maketitle

  \begin{abstract}
  A forbidden transition graph is a graph defined together with a set of permitted transitions \textit{i.e.} unordered pair of adjacent edges that one may use consecutively in a walk in the graph. In this paper, we look for the smallest set of transitions needed to be able to go from any vertex of the given graph to any other. We prove that this problem is NP-hard and study approximation algorithms. We develop theoretical tools that help to study this problem.
  \end{abstract}
	
\section{Introduction}
\label{sec:intro}

Graphs are the model of choice to solve routing problems in all sorts of networks. Depending on the applications, we sometimes need to express stronger constraints than what the standard definitions allow for. Indeed, in many practical cases, including optical networks, road networks or public transit systems among others, the set of possible walks a user can take is much more complex than the set of walks in a graph (see \cite{Ahmed} or \cite{SCTM} for examples). To model a situation where a driver coming from a given road may not turn left while both the road he comes from and the road on the left exists, we have to define the permitted walks by taking into account not only the edges of the graph that a walk may use but also the transitions. A transition is a pair of adjacent edges and we call forbidden-transition graph a graph defined together with a set of permitted transitions. 

Graphs with forbidden transitions have appeared in the literature in \cite{Kotzig} and have received a lot of interest since, as well as other more specific models such as properly colored paths \cite{Daykin,Gutin}. 
Many problems are harder in graphs with forbidden transitions, such as determining the existence of an elementary path (a path that does not use twice the same vertex) between two vertices which is a well-known polynomial problem in graph without forbidden transitions and has been proved NP-complete otherwise (\cite{szeider}). Algorithms for this problem have been studied in the general case \cite{Kante-algopath} and also on some subclasses of graphs \cite{Kante-grid}.

Forbidden transitions can also be used to measure the robustness of graph properties. In \cite{Sudakov}, Sudakov studies the Hamiltonicity of a graph with the idea that even Hamiltonian graphs can be more or less strongly Hamiltonian (an Hamiltonian graph is a graph in which there exists an elementary cycle that uses all the vertices). 
The number of transitions one needs to forbid for a graph to lose its Hamiltonicity gives a measure of its robustness: if the smallest set of forbidden transitions that makes a graph lose its Hamiltonicity has size 4, this means that this graph can withstand the failure of three transitions, no matter where the failures happen.

The notion we are interested in in this paper is not Hamiltonicity but connectivity (the possibility to go from any vertex to any other), which is probably one of the most important properties we expect from any communication or transport network. 
However, our work differs from others in that we are not looking for the minimum number of transitions to forbid to disconnect the graph but for the minimum number of transitions to allow to keep the graph connected (the equivalent of minimum spanning trees for transitions). 
In other words, we are looking for the maximum number of transitions that can fail without disconnecting the graph, provided we get to choose which transitions still work. 
This does not provide a valid measure of the robustness of the network but measuring the robustness is only one part (the definition of the objective function) of the problem of robust network design. 
In most practical situations, robustness is achievable but comes at a cost and the optimization problem consists in creating a network as robust as possible for the minimum cost. In this respect, it makes sense to be able to choose where the failure are less likely to happen. Our problem highlights which transitions are the most important for the proper functioning of the network and this is where special attention must be paid in its design or maintenance. As long as those transitions work, connectivity is assured.

We also would like to point out that in practice, unusable transitions are not always the result of a malfunction. Consider a train network and imagine that there is a train going from a town $A$ to a town $B$ and one going from the town $B$ to a town $C$. In the associated graph, there is an edge from $A$ to $B$ and one from $B$ to $C$ but if the second train leaves before the first one arrives, the transition is not usable and this kind of situation is clearly unavoidable in practice even if no special problem happens. 
Highlighting the most important transitions in the network thus helps design the schedule, even before the question of robustness arises.

Unlike Hamiltonicity or the existence of elementary path between two vertices, testing the connectivity is an easy task to perform even on graphs with forbidden transitions (note that a walk connecting two vertices does not have to be elementary). 
However, we prove that the problem of determining the smallest set of transitions that maintains the connectivity of the given graph is NP-hard even on co-planar graph which is the main contribution of the paper (see Section \ref{sec:NPC}). 
Other notable contributions include a $O(|V|^2)$-time $\frac 3 2$-approximation (Theorem \ref{th:aprox}) and a reformulation of the problem (Theorem \ref{thm:equipb}) which was of great help in the proofs of the other results and could hopefully be useful again in subsequent works.	

\paragraph{\textbf{Definitions and notations}}
\label{sec:def}
Throughout this paper, we only consider finite simple graphs, \textit{i.e.} undirected graphs with a finite number of vertices, no multiple edges and no loop.
Let $G$ be a graph. The vertex set of $G$ is denoted by $V(G)$ and its edge set by $E(G)$. 
The size of a set $S$ is denoted by $|S|$. We denote by $d(v)$ the degree of a vertex $v$.
We write $xy$ to denote an edge $\{x,y\}$. 
We define a walk in $G$ as a sequence $W=(v_1,\dots,v_k)$ of vertices such that for all $i\leqslant k-1$, $v_iv_{i+1}\in E(G)$ and we say that $W$ uses the edge $v_iv_{i+1}$. 
Here, we say that the walk $W$ leads from the vertex $v_1$ to $v_k$.

For $X\subseteq V(G)$, we denote by $G[X]$ the subgraph of $G$ induced by $X$. 
We also denote by $G-X$ the subgraph of $G$ induced by $V(G)\setminus X$. For $x\in V(G)$, we write $G-x$ instead of $G-\{x\}$.
We denote by $\bar G$ the complement of $G$ \ie the graph such that $V(\bar G)=V(G)$ and $E(\bar G)=\{xy\in{V(G)\choose 2} : xy\notin E(G)\}$. 
We say that a graph $G$ is co-connected if and only if $\bar G$ is connected. 
We also call \emph{co-connected components} (or \emph{co-cc}) of $G$ the connected components of $\bar G$.

\paragraph{\bf Transitions} A transition is a set of two adjacent edges. 
We write $abc$ for the transition $\{ab,bc\}$. 
If a walk uses the edges $ab$ and $bc$ consecutively (with $a\neq c$), we say that it uses the transition $abc$. For example, the walk $(u,v,w,v,x)$ uses the transitions $uvw$ and $wvx$. 
Let $T$ be a set of transitions of $G$ and $W=(v_1\dots v_k)$ be a walk on $G$. 
We say that $W$ is $T$-compatible if and only if it only uses transitions of $T$ \ie for all $i\in[1,k-2]$, we have  $v_i v_{i+1} v_{i+2} \in T$ or $v_i=v_{i+2}$ (\ie $v_iv_{i+1}$ and $v_{i+1}v_{i+2}$ are the same edge). 
Observe that a walk consisting of two vertices is always $T$-compatible. 
If for all vertices $u$ and $v$ of $V(G)$, there exists a $T$-compatible walk between $u$ and $v$, then we say that $G$ is $T$-connected and that $T$ is a connecting transition set of $G$. 
The problem we study here is the following:
\pbDef{Minimum Connecting Transition Set (MCTS)}
{A connected graph $G$.}
{A minimum connecting transition set of $G$.}

	\section{Polynomial algorithms and structural results}
	\label{sec:result}
	In this section, we only consider graphs with at least 2 vertices. Our problem is trivial otherwise.
	
	\begin{lemma}\label{lem:tree}
		If $G$ is a tree then a minimum connecting transition set of $G$ has size $|V(G)|-2$.
	\end{lemma}
	\begin{proof}
		We first prove that $|V(G)|-2$ transitions are enough to connect $G$. For every vertex $v$ of $G$, we pick a neighbor of $v$ that we call $f(v)$. For every neighbor $u\neq f(v)$ of $v$, we allow the transition $uvf(v)$. We end up with the transition set $T=\{uvf(v) : v\in V(G), u\in N(v)\setminus \{f(v)\}\}$.
		Let $u$ and $v$ be vertices of $G$. Since $G$ is connected, there exists a walk $(u,u_1,u_2,\dots,u_k,v)$. The walk $(u,u_1,f(u_1),u_1,u_2,f(u_2),u_2,\dots,u_k,f(u_k),u_k,v)$ is $T$-compatible and still leads from $u$ to $v$. This proves that $G$ is $T$-connected.
		The size of $T$ is $|T|=\sum_{v\in V(G)} (d(v)-1)=2|E(G)|-|V(G)|$. Since $G$ is a tree, $|E(G)|=|V(G)|-1$ and thus, $|T|=|V(G)|-2$.
		
		Let us now prove by induction on the number $n$ of vertices of $G$ that at least $n-2$ transitions are necessary to connect $G$. 
		This is obvious for $n=2$. Let us assume that it holds for $n$ and let $G$ be a tree with $n+1$ vertices. 
		Let $T$ be a minimum connecting transition set of $G$.
		Let $uv$ be an internal edge of $T$ if any (\ie an edge such that $u$ and $v$ are not leaves). 
		Let $a$ and $b$ be two vertices from different connected components of $G-\{u,v\}$.
		Every walk leading from $a$ to $b$ in $G$ therefore uses the edge $uv$ and thus, two transitions containing $uv$. 
		This proves that every internal edge of $T$ belongs to at least two transitions of $T$. 
		If every edge of $G$ belongs to at least two transitions of $T$, $T$ has size at least $|E(G)|=|V(G)|-1$ which concludes the proof. 
		Otherwise, let $uv$ be an edge that belongs to at most one transition of $T$. This means that one of its vertices, say $v$, is a leaf. 
		It is straightforward to check that $uv$ must belong to one transition of $T$, otherwise $G$ would not be $T$-connected. Let $t$ be the transition in $T$ containing $uv$.
		The graph $G-v$ is $T\setminus\{t\}$-connected and is a tree. By the induction hypothesis, this means that $|T\setminus\{t\}|\geqslant n-3$ and $|T|\geqslant n-2$. This concludes the proof of the lemma.
		\qed
	\end{proof}
	
	Let us also note that a linear-time algorithm to compute an optimal solution can be easily deduced from this proof.
	Since every connected graph contains a spanning tree, we have the following corollary. 
	\begin{corollary}\label{cor:n-2}
		Every connected graph $G$ has a connecting transition set of size $|V(G)|-2$.
	\end{corollary}
	
	Note however that in the general case, this bound is far from tight.
	The most extreme case is the complete graph where every vertex can be connected to every other with a walk of one edge, that therefore uses no transition. 
	Thus, the empty set is a connecting transition set of the complete graph. The following result aims at tightening the upper bound on the size of the minimum connecting transition set of a graph.
	
	\begin{theorem}\label{thm:tau}
		Every connected graph $G$ has a connecting transition set of size $\tau(G)$ where
		
		\[\tau(G)=
		\dsum{\underset{|C|\geqslant 2}{C\text{ }\mathrm{co}\text{-}\mathrm{cc}\text{ }\mathrm{of}\text{ }G}}{} 
		\left\{
		\begin{split}
		|C|-2&\text{ if the subgraph of }G\text{ induced by }C\text{ is connected}\\
		|C|-1&\text{ otherwise }
		\end{split}
		\right.\]
	\end{theorem}
	
	\begin{proof}
		By definition, if $u$ and $v$ belong to different co-connected components of $G$, there is an edge $uv\in E(G)$ and there is therefore a walk between $u$ and $v$ is compatible with any transition set. 
		We only have to find a transition set that connects all the vertices that belong to the same co-connected component.
		
		Let $C$ be a co-connected component of $G$ with at least 2 vertices. 
		If $G[C]$ is connected, Corollary \ref{cor:n-2} provides a transition set of size $|C|-2$ that connects $C$. Otherwise, since $G$ is connected, we know that $V(G)\neq C$ and there exists a vertex $v\notin C$. 
		Hence, $v$ is adjacent to every vertex of $C$ and $C\cup\{v\}$ induces a connected subgraph of $G$. Corollary \ref{cor:n-2} provides a set of size $|C\cup\{v\}|-2=|C|-1$ that connects $C$. 
		By iterating this on every $C$, we build a connecting transition set $T$ of size $\tau(G)$.
	\qed
	\end{proof}
	
	Note that this bound can be computed in $O(|V(G)|^2)$. 
	However, this bound is still not tight. 
	Let us consider the graph $\bar{P_7}$ whose vertex set is $\{v_1,\dots,v_7\}$ and where every vertex $v_i,\, 2\leqslant i \leqslant 6$ is connected to every vertex of the graph but $v_{i-1}$ and $v_{i+1}$. 
	Since the graph is connected and co-connected, $\tau(\bar{P_7})=5$ but the set $T=\{v_3v_1v_4, v_2v_4v_1, v_6v_4v_7, v_5v_7v_4\}$ is a connecting transition set of size only 4. 
	To better understand this solution, let us consider the spanning tree of $\bar{P_7}$ depicted in Figure \ref{spantree}:
	
	\begin{figure}[h]
\begin{center}
	\begin{pspicture}(2.4,1.9)
		
		\psline{-}(0.2,0.2)(0.7,0.7)
		\psline{-}(1.2,1.2)(0.7,0.7)
		\psline{-}(1.2,1.2)(1.7,0.7)
		\psline{-}(1.7,0.7)(2.2,0.2)
		\psline{-}(1.2,1.2)(1.7,1.7)
		\psline{-}(1.2,1.2)(0.7,1.7)
		
		\pscircle[linecolor=black, fillstyle=solid, fillcolor=white](0.2,0.2){0,25}
		\pscircle[linecolor=black, fillstyle=solid, fillcolor=white](0.7,0.7){0,25}
		\pscircle[linecolor=black, fillstyle=solid, fillcolor=white](1.2,1.2){0,25}
		\pscircle[linecolor=black, fillstyle=solid, fillcolor=white](1.7,0.7){0,25}
		\pscircle[linecolor=black, fillstyle=solid, fillcolor=white](2.2,0.2){0,25}
		\pscircle[linecolor=black, fillstyle=solid, fillcolor=white](0.7,1.7){0,25}
		\pscircle[linecolor=black, fillstyle=solid, fillcolor=white](1.7,1.7){0,25}
		
		\rput(0.2,0.2){$v_3$}
		\rput(0.7,0.7){$v_1$}
		\rput(1.2,1.2){$v_4$}
		\rput(1.7,0.7){$v_7$}
		\rput(2.2,0.2){$v_5$}
		\rput(0.7,1.7){$v_2$}
		\rput(1.7,1.7){$v_6$}
		
		\end{pspicture}
		\caption{A spanning tree of $\bar{P_7}$.}\label{spantree}
		\end{center}
\end{figure}
	Note that the set $T$ described above does not connect this spanning tree. 
	Indeed, one can not go from $v_1$, $v_2$ or $v_3$ to $v_5$, $v_6$ or $v_7$ using a $T$-compatible walk in the tree. 
	However, these vertices are already connected to each other by edges that do not belong to the spanning tree. 
	The optimal solution here does not consist in connecting a spanning tree of $G$ but in connecting  a spanning tree of $G[\{v_1,v_2,v_3,v_4\}]$ and one of $G[\{v_4,v_5,v_6,v_7\}]$ and the cost is $(4-2)+(4-2)=4$ instead of $7-2=5$. 
	
	In fact, we will prove that to each optimal connecting transition set $T$ of a graph $G$ corresponds an unique decomposition of $G$ into subgraphs $G_1,G_2,\dots,G_k$ such that $T$ is the disjoint union of $T_1,T_2,\dots,T_k$, where each $T_i$ is the connecting transition set of some spanning tree of $G_i$.
	Observe that the size of $T$ is uniquely determined by its correspondent decomposition, \ie,  $|T| = |V(G_1)| -2 + \dots + |V(G_k)| -2$.
	Hence, finding an optimal connecting transition set is equivalent to finding its correspondent decomposition.
	In the following, we reformulate \textbf{MCTS} into this problem of graph decomposition which is easier to work with.

	\begin{definition}{Connecting Hypergraph}
	
		Let $G$ be a graph. A connecting hypergraph of $G$ is a set $H$ of subsets of $V(G)$, such that
		\begin{itemize}
			\item For all $E\in H$,  we have $G[E]$ is connected and $|E|\geqslant 2$.
			\item For all $uv \notin E(G)$, there exists $E\in H$ such that $u,v \in E$ (we say that the hyperedge $E$ connects $u$ and $v$).
		\end{itemize}
	\end{definition}
	
	We define the problem of optimal connecting hypergraph as follows:
	\pbDef{Optimal Connecting HyperGraph (OCHG)}
	{A connected graph $G$.}
	{A connecting hypergraph $H$ that minimizes $\mathrm{cost}(H)=\sum_{E\in H}{(|E| - 2) }$.}
	
	In the next theorem, we prove that \textbf{OCHG} is a reformulation of \textbf{MCTS}.

	\begin{theorem}\label{thm:equipb}
		Let $G$ be a graph.
		\begin{itemize}
			\item The size of a minimum connecting transition set of $G$ is the same as the cost of an optimal connecting hypergraph.
			\item A solution of one of these problems on $G$ can be deduced in polynomial time from a solution of the other.
		\end{itemize}	
	\end{theorem}
	\begin{proof}
		Let $G$ be a graph. This theorem is implied by the two following claims.
		\begin{claim}\label{claim:equi1}
			Let $H=\{E_1,\dots,E_k\}$ be a connecting hypergraph of $G$. There exists a connecting transition set $T$ of size  at most $\mathrm{cost}(H)$.
		\end{claim}
			By the definition of a connecting hypergraph, each $E_i$ induces a connected graph and by Corollary \ref{cor:n-2}, there exists a subset of transitions $T_i$ of size $|E_i|-2$ such that $G[E_i]$ is $T_i$-connected. 
			Let $T=\bigcup_{i\leqslant k}T_i$. 
			By definition, for all $uv\notin E(G)$, there exists $i$ such that $u,v\in E_i$. 
			Since $G[E_i]$ is $T_i$-connected and $T_i\subseteq T$, there is a $T$-compatible walk between $u$ and $v$ in $G$ which means that $G$ is $T$-connected. 
			Since $T=\bigcup_{i\leqslant k}T_i$, $|T|\leqslant\sum_{i\leqslant k}|T_i|=\sum_{i\leqslant k}(|E_i|-2)=\cost(H)$.

		\begin{claim}\label{claim:equi2}
			Let $T$ be a connecting transition set of $G$. There exists a  connecting hypergraph $H=\{E_1,\dots,E_k\}$ of cost at most $|T|$.
		\end{claim}
			Let $\sim$ be the relation on $T$ such that $t\sim t'$ if $t$ and $t'$ share at least one common edge. 
			We denote by $\cR$ the transitive closure of $\sim$. 
			Let $T_1,\dots, T_{k}$ be the equivalence classes of $\cR$. 
			For all $i\leqslant k$, we denote by $E_{i}$ the set of vertices induced by $T_i$. 
			We claim that the hypergraph $\{E_1,\dots,E_k\}$ is a connecting hypergraph and that, for all $i$, $|T_i|\geqslant |E_i|-2$. 
			
			By construction, for all $i$, we have $|E_i|\geqslant 3$ since $T_i$ contains at least one transition and thus, three vertices.
			Furthermore, since $G$ is $T$-connected, there exists a $T$-compatible walk $W$ between every pair $uv\notin E(G)$. All the transitions that $W$ uses must be in $T$ and are pairwise equivalent for $\cR$. 
			Thus, for all $uv\notin E(G)$, there exists $i$ such that both $u$ and $v$ belong to $E_i$.
			
			It remains to prove that for all $i$, $|E_i|-2\leqslant |T_i|$. 
			We prove by induction on $n$ that every set $T$ of $n$ pairwise equivalent transitions induces a vertex set of size at most $n+2$. This property trivially holds for $n=1$.
			Now, suppose that it is true for sets of size $n$ and let $T$ be a set of pairwise equivalent transitions of size $n+1$.  
			Let $P=t_1,\dots,t_r$ be a maximal sequence of distinct transitions of $T$ such that, for all $i\leqslant r-1$, $t_i \sim t_{i+1}$. 
			One can check that all the transitions of $T\setminus \{t_1\}$ are still pairwise equivalent 
			(otherwise, $P$ would not be maximal).
			By the induction hypothesis, $T\setminus \{t_1\}$ induces at most $n+2$ vertices. 
			Since $t_1$ shares an edge (and thus at least 2 vertices) with $t_2$, it induces at most one vertex not induced by $T\setminus \{t_1\}$. 
			Thus $T$ induces at most $n+3$ vertices.
	\qed
	\end{proof}

	Let us note that the bound provided in Theorem \ref{thm:tau} suggests a $O(|V|^2)$-time heuristic for \textbf{OCHG} which consists in building the set $H$ as follows: 
	
	\[H=
	\dcup{\underset{|C|\geqslant 2}{C\text{ }\mathrm{co}\text{-}\mathrm{cc}\text{ }\mathrm{of}\text{ }G}}{} 
	\left\{
	\begin{split}
	C&\text{ if the subgraph of }G\text{ induced by }C\text{ is connected}\\
	C&\cup \{v\}\text{ with }v\notin C\text{ otherwise }
	\end{split}
	\right.\]
	
	We use the reformulation given by Theorem \ref{thm:equipb} to generalize Lemma \ref{lem:tree}:
	
	\begin{lemma}\label{lem:ptsarticu}
		If $G$ has a cut vertex, then a minimum connecting transition set of $G$ has size $|V(G)|-2$.
	\end{lemma}
	
	\begin{proof}
		By Theorem \ref{thm:equipb}, it is sufficient to prove that $H=\{V(G)\}$ is an optimal connecting hypergraph of $G$. 
		Let $p$ be a cut vertex of $G$ and $C_1,\dots,C_r$ be the connected components of $G-p$. 
		Let $H=\{E_1,\dots,E_k\}$ be an optimal connecting hypergraph of $G$. 
		
		Let $a\in C_1$. Suppose that there are two vertices $b,c\neq p$ that do not belong to $C_1$. 
		Hence, $\{a,b\}\notin E(G)$ and there exists $i$ such that $a,b\in E_i$. 
		Since $E_i$ must induce a connected subgraph of $G$, we know that $p\in E_i$. 
		Similarly, we know that there exists $E_j$ that contains $a,c$ and $p$. 
		Thus, we have $|E_i\cap E_j|\geqslant |\{a,p\}|\geqslant 2$ and $\cost(\{E_i\cup E_j\})= |E_i\cup E_j| - 2 \leqslant |E_i| -2 + |E_j| -2=\cost(\{E_i,E_j\})$. 
		
		Thus, $H\setminus \{E_i,E_j\}\cup \{E_i\cup E_j\}$ is also an optimal connecting hypergraph where the same hyperedge contains both $b$ and $c$. 
		By iterating this process, we prove that there is an optimal connecting hypergraph with one hyperedge $E$ that contains $a$, $p$ and $C_2,\dots,C_r$. This result trivially holds if there is only one vertex $b\neq p$ that does not belong in $C_1$.
		By iterating the previous process on this hypergraph with a vertex in $E\cap C_2$ instead of $a$, we end up with the optimal connecting hypergraph  $ \{V(G)\}$ whose cost is $n-2$.
	\qed
	\end{proof}
	
	The following lemma will help us to prove that \textbf{MCTS} admits a $\frac{3}{2}$-approximation and its NP-hardness.
	It proves that if the graph is co-connected, we can restrict ourselves to some specific connecting hypergraph.
	
	\begin{lemma}\label{lem:co-connected}
		Let $G$ be a connected graph. If $G$ is co-connected or $G$ has a dominating vertex $x$ and $G-x$ is connected and co-connected,
		then there exists an optimal connecting hypergraph $H=\{E_1,\dots,E_k\}$ on $G$ such that for all $i$, $G[E_i]$ is co-connected.
	\end{lemma}
	
	\begin{proof}
	 	
	 	Let $H$ be an optimal connecting hypergraph on $G$ and let $E$ be an hyperedge of $H$ that is not co-connected. If $G[E]$ is complete, then $E$ does not connect any pair of non-adjacent vertices and $H\setminus E$ is still a connecting hypergraph whose cost is less or equal than the cost of $H$. Else, let $a$ and $b$ be two non-adjacent vertices of $E$. They therefore belong to the same co-connected component $C$ of $G[E]$.
	 	
	 	If $C$ is a co-connected component of $G$, then, since $C\subsetneq E\subsetneq V(G)$, we know that $G$ is not co-connected and by hypothesis, it therefore has a dominating vertex $x$ and $C=V(G)\setminus \{x\}$ and thus, $E=V(G)$. 
	 	Hence, $\cost(H)\geqslant |V(G)|-2$ which is absurd since $\{V(G)\setminus \{x\}\}$ is a connecting hypergraph of cost $|V(G)|-3$.
	 	
	 	Thus, $C$ is not a co-connected component of $G$, which means that there exists $u\in C$ and $v\in V(G)\setminus C$ such that $u$ and $v$ are not adjacent. 
	 	To facilitate the understanding, the construction we use in this case is illustrated in Figure \ref{fig:co-connected}.
	 	\begin{figure}[h!]
	 		\centering
	 		\includegraphics[width=0.8\columnwidth]{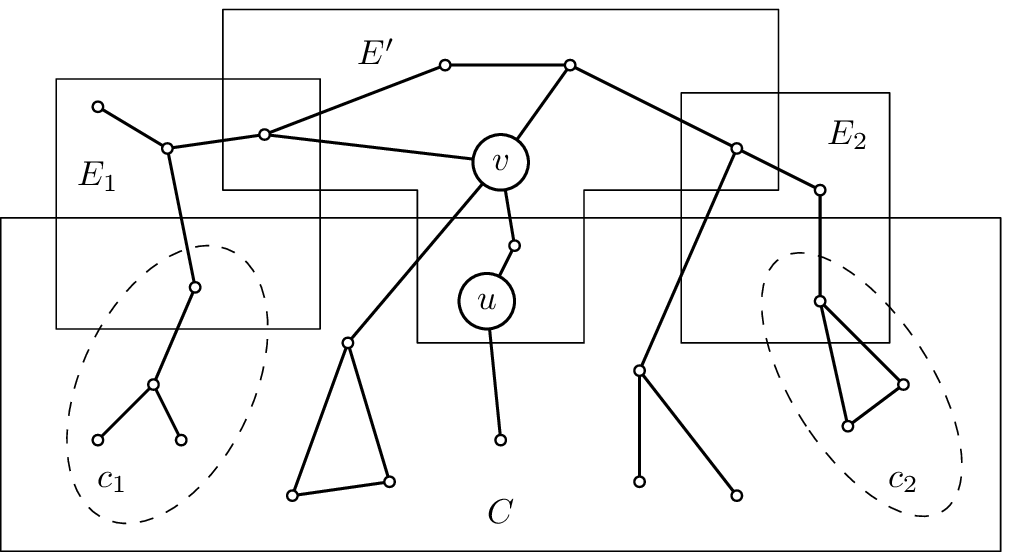}
	 		\caption{Here, $G[C]$ has five connected components, two of which ($c_1$ and $c_2$) are not connected to $E'$. To ensure that $\mathscr E$ is connected, we need the hyperedges $E_1$ and $E_2$.}
	 		\label{fig:co-connected}
	 	\end{figure}
	 	Since $C$ is a co-cc of $G[E]$, we also know that $v\notin E$. Hence, there exists $E'\neq E$ in $H$ that contains $u$ and $v$. 
	 	Let $c_1,\dots,c_l$ be the connected components of $G[C]$ that are not connected to any vertex of $E'$ (if any). 
	 	By definition of connecting hypergraph, we know that for all $i\leqslant l$, there exists $E_i\in H$ that connects a vertex of $E'$ and a vertex of $c_i$. We create $H'$ from $H$ by replacing $E$ by $E\setminus C$ and by replacing $E_1,\dots,E_l$ and $E'$ by $\mathscr E=E'\cup C \cup E_1 \cup \cdots \cup E_l$.
	 	
	 	We claim that $cost(H')\leqslant cost(H)-1$. Indeed, replacing $E$ by $E\setminus C$ decreases the cost by $|C|$ while replacing $E'$ by $E'\cup C$ increases the cost of at most $|C|-1$ because $E'\cap C$ contains at least the vertex $u$. Moreover, we can prove by induction on $i\leqslant l$ that the cost of $\{E'\cup C \cup E_1 \cup \cdots \cup E_{i-1},E_i\}$ is greater or equal than the cost of $\{E'\cup C \cup E_1 \cup \cdots \cup E_{i-1}\cup E_i\}$. Indeed, $(E'\cup C \cup E_1 \cup \cdots \cup E_{i-1})\cap E_i$ has size at least two (it contains at least one vertex in $E'$ and one in $c_i\subset C$, by definition of $E_i$). It follows that $\cost (\mathscr E)\leqslant \cost (\{E'\cup C, E_1, \dots, E_l\})$. Therefore, $\cost(H')\leqslant \cost(H)-1$. Since $H$ is an optimal connecting hypergraph, we know that $H'$ is not a connecting hypergraph. But observe that $H'$ satisfies the following properties:
	 	\begin{itemize}
	 		\item Every pair of non-adjacent vertices is still connected by an hyperedge of $H'$. Indeed, if two non-adjacent vertices are connected by $E$ in $H$ they are connected by $E\setminus C$ or $\mathscr E$ in $H'$ depending on whether they belonged to $C$ or not; if they are connected by an $E_i$ in $H$, they are connected by $\mathscr E$ in $H'$ and otherwise, the hyperedge that connects them in $H$ belongs to $H'$ too.
	 		\item The graph $G[\mathscr E]$ is connected. Indeed, the sets $E',E_1,\dots,E_l$ all induce connected subgraphs of $G$ by definition and are connected to each other because for all $i$, $E_i\cap E'\neq\varnothing$. Furthermore, all the connected components of $C$ are connected to a vertex of $E'$, except the $c_i$ which are by definition connected to the $E_i$.  
	 	\end{itemize}
	 	
	 	Thus, either $|E\setminus C|<2$ or $G[E\setminus C]$ is not connected.
	 	If $E\setminus C$ is a singleton, it does not connect any pair of non-adjacent vertices. Thus, $H'\setminus\{E\setminus C\}$ is a connecting hypergraph whose cost is strictly smaller than $H$, which is absurd. Hence, $G[E\setminus C]$ is not connected, which means it is co-connected. We can therefore apply to $E\setminus C$ the same method we used on $C$.
	 	
	 	Just like before, we know that there exists two non-adjacent vertices $u'\in E\setminus C$ and $v'\notin E$. Let $F\in H'$ be the hyperedge that connects $u'$ and $v'$, let $d_1,\dots,d_{l'}$ be the connected components of $E\setminus C$ that are not connected to $F$ and let $F_1,\cdots,F_{l'}$ be hyperedges of $H'$ such that $F_i$ connects a vertex of $F$ to a vertex of $d_i$. We create $H''$ from $H'$ by removing $E\setminus C$ and by replacing $F_1,\dots,F_{l'}$ and $F$ by $\mathscr F=F\cup (E\setminus C) \cup F_1 \cup \cdots \cup F_{l'}$.
	 	
	 	With the same arguments used for $H'$, we can prove that $\mathscr F$ is connected and that $H''$ connects every pair of non-adjacent vertices, which means that $H''$ is a connecting hypergraph. Moreover, with these arguments, we can also prove that $\cost(\mathscr F)\leqslant \cost(\{F\cup (E\setminus C) , F_1 , \dots , F_{l'}\})$. Furthermore, removing $E\setminus C$ from $H'$ decreases the cost by $|E\setminus C|-2$ and replacing $F$ by $F\cup (E\setminus C)$ increases it by at most $|E\setminus C|-1$ since $F\cap (E\setminus C)$ contains at least the vertex $u'$. Hence, $\cost(H'')\leqslant \cost(H')+1\leqslant \cost(H)$ and then $H''$ is an optimal covering hypergraph. Observe that $H''$ has strictly fewer hyperedges $E$ such that $G[E]$ is not co-connected than $H$. We prove the lemma by iterating this process.
	 \qed
	\end{proof}

%
	
	We now prove that \textbf{MCTS} has a polynomial $\frac 3 2$-approximation:

	\begin{theorem}
		\label{th:aprox}
		For every connected graph $G$ and optimal connecting transition set $T$ of $G$, the size of $T$ is at least $2/3\tau(G)$, where $\tau(G)$ is the function defined in Theorem \ref{thm:tau}.
	\end{theorem}
	
	\begin{proof}
		By Theorem \ref{thm:equipb}, it is enough to prove that an optimal connecting hypergraph has cost at least $2/3\tau(G)$. 
		We start by proving the following claim which proves the theorem on the graphs that are connected and co-connected.
		\begin{claim}\label{claim:2/3tau}
			Let $G$ be a connected and co-connected graph with $n$ vertices. For every connecting hypergraph $H$ of $G$, we have  $\cost(H)\geq \frac{2(n-1)}{3}$.
		\end{claim}
		\begin{proof}
			We know by Lemma \ref{lem:co-connected} that there exists an optimal connecting hypergraph $H=\{E_1,\dots,E_k\}$ of $G$ such that for all $i$, $G[E_i]$ is co-connected. 
			
			First, observe that $\cost(H)\geq 2k$. 
			Indeed, for every $i$, $G[E_i]$ is both connected and co-connected, thus we have $|E_i|\geq 4$. As $\cost(H)=\sum_{i\leqslant k}|E_i| -2$, we deduce that $\cost(H)\geq 2k$.
			
			Now, we prove that $\cost(H)\geq n- k-1$.
			Observe that since $\bar G$ is connected, every vertex $v$ belongs to at least one edge in $\bar G$. Hence, by definition of connecting hypergraph, there exists $E\in H$ such that $v\in E$ and ${\bigcup}_{i\leqslant k} E_i=V(G)$.
			
			Also note that for all $i< k$, there exists an hyperedge $E_j$ with $j>i$ that shares a vertex with an hyperedge of $E_1,\dots,E_i$. Otherwise, ${\bigcup}_{j\leqslant i} E_j$ and $\bigcup_{j>i}E_j$ cover the vertices of $G$ and since $G$ is co-connected, this means that there exists $u\in{\bigcup}_{j\leqslant i} E_j$ and $v\in{\bigcup}_{j>i}E_j$ such that $(u,v)\notin E(G)$ but no set of $H$ connects them, which is impossible. We can assume without loss of generality that this hyperedge that shares at least one vertex with ${\bigcup}_{j\leqslant i} E_j$ is $E_{i+1}$.
			
			It is now immediate to prove by induction on $i\leqslant k$ that ${\sum}_{j\leqslant i}|E_j|-2\geqslant |{\bigcup}_{j\leqslant i}E_j| -i -1$. Thus, we have $\cost(H)\geq n - k- 1$.
			
			By combining the two inequalities, we find that $\cost(H)\geq \frac{2(n-1)}{3}$.
		\qed
	\end{proof}
		
		Let $H=\{E_1,\dots,E_k\}$ be an optimal connecting hypergraph of $G$, let $C_1,\dots, C_l$ be the co-connected components of $G$ and for all $j\leqslant l$, let $v_j$ be a vertex that does not belong to $C_j$.
		
		For all $i\leqslant k$ and $j\leqslant l$ such that $|E_i\cap C_j|\geqslant 2$, we define 
		\[E_{i,j}=\left\{\begin{split}
		E_i\cap C_j\quad\quad&\text{ if }G[E_i\cap C_j]\text{ is connected}\\
		E_i\cap C_j\cup\{v_j\}\quad&\text{ otherwise}
		\end{split}\right.\]
		and we define $F$ as the union of the $\{E_{i,j}\}$. Note that if $G$ is co-connected, there is only one co-connected component $C_1=V(G)$ and while there is no vertex $v_1\notin C_1$, for all $i$, $G[E_i\cap C_1]=G[E_i]$ is connected by definition, so we do not need $v_1$ in the above construction.
		
		Since $v_j$ dominates $C_j$, it is easy to check that every hyperedge of $F$ is connected and has size at least 2. Plus, any two non-adjacent vertices of $G$ belong to the same $C_j$ and are connected by an $E_i\in H$. Therefore, they belong to $E_i\cap C_j$ which has size at least two and thus belongs to $F$. Hence, $F$ is a connecting hypergraph.
		
		Let $E_i\in H$ and let $S_i$ be the set of values of $j$ such that $E_{i,j}$ exists. If there is only one such value $j$, then $\cost(E_i)\geqslant \cost(E_{i,j})=\cost(\bigcup_{j\in S_i}\{E_{i,j}\})$ follows immediately. Otherwise
		\[\cost(\bigcup_{j\in S_i} \{E_{i,j}\})\leqslant \sum_{j\in S_i}(|E_i \cap C_j|+|\{v_j\}|-2) \leqslant |E_i|-|S_i| \leqslant |E_i|-2=\cost(E_i)\]
		still holds. Since $F=\bigcup_{i\leqslant k}\bigcup_{j\in S_i} \{E_{i,j}\}$, it follows that $\cost(F)\leqslant \cost(H)$ which proves that $F$ is optimal.
		
		We know that two non-adjacent vertices necessarily belong to the same co-connected component of $G$ and since an hyperedge $E_{i,j}$ only contains one vertex that does not belong to $C_j$ it only connects non-adjacent vertices of one connected component. 
		For all $j$, let $F_j$ be the set of hyperedges of $F$ that connect non-adjacent vertices of $C_j$. 
		Since $E_{i,j}\subseteq C_j\cup \{v_j\}$, $F_j$ is a connecting hypergraph of $G[C_j\cup\{v_j\}]$. 
		
		Let $C_j$ be a co-connected component of $G$ and observe that:
		 \begin{itemize}
		 	\item if $C_j$ is not connected, then $v_j$ is a cut vertex of $C_j\cup\{v_j\}$. Hence, by Lemma \ref{lem:ptsarticu}, $\cost{(F_j)}\leqslant |C_j|-1$.
		 	
		 	\item  if $C_j$ is connected, since it is co-connected by definition, $C_j\cup \{v_j\}$ admits by Lemma \ref{lem:co-connected} an optimal connecting hypergraph $H_j$ such that for every hyperedge $E$ of $H_j$, $G[E]$ is co-connected and thus, $v_j\notin E$. 
		 	This proves that $H_j$ is a connecting hypergraph of $G[C_j]$. 
		 	As $G[C_j]$ is connected and co-connected, we know by the above claim that $\cost(H_j)\geq \frac{2(|C_j|-1)}{3}$. 
		 \end{itemize}
		 Thus, we have
		\[\begin{split}\cost(F)=\sum_{j\leqslant l}\cost(F_j)&\geqslant \sum_{j\leqslant l}\left\{\begin{split}
		|C_j|-1\quad&\text{ if }C_j\text{ is not connected}\\
		\frac {2(|C_j|-2)} 3&\text{ otherwise}
		\end{split}\right.\\
		&\geqslant \frac 2 3\sum_{j\leqslant l}\left\{\begin{split}
		|C_j|-1&\text{ if }C_j\text{ is not connected}\\
		|C_j|-2 &\text{ otherwise}
		\end{split}\right.\\
		&\geqslant \frac 2 3 \tau(G)
		\end{split}\]
	\qed
	\end{proof}

We also can prove that this bound is tight. Indeed consider the graph $G$ defined as the complement of a star of $n$ branches of 3 edges. The graph $G$ has $3n+1$ vertices that we call $c$, $v_{i,1}$, $v_{i,2}$ and $v_{i,3}$ with $1\leqslant i \leqslant n$ (in $\bar G$, $c$ is the center of the star and $v_{i,1}$, $v_{i,2}$ and $v_{i,3}$ are the three vertices of the branch $i$). Every vertex of $G$ is connected to every other except $c$ and $v_{i,1}$, $v_{i,1}$ and $v_{i,2}$ and $v_{i,2}$ and $v_{i,3}$ with $1\leqslant i \leqslant n$. Since $G$ is both connected and co-connected, our algorithm returns the connecting hypergraph $H_1=\{V(G)\}$ whose cost is $3n-1$ but the hypergraph $H_2=\cup_{1\leqslant i \leqslant n} \{c,v_{i,1}, v_{i,2}, v_{i,3}\}$ is a connecting hypergrpah of cost $2n$. The example of co-$P_7$ that we used to prove that the algorithm was not exact (see Figure \ref{spantree}) is the case $i=2$.

	\section{NP-hardness}\label{sec:NPC}
	
	In this section, we give a proof of NP-hardness of \textbf{OCHG} which involves very dense graphs. 
	Hence, we prefer to work with the complementary graphs and therefore prove the NP-hardness of the following problem that we call \textbf{co-OCHG}:
	
	\begin{definition}{co-Connecting Hypergraph}
		Let $G$ be a graph. 
		A co-connecting hypergraph is a collection of  hyperedges $E_1,\dots,E_r\subseteq V(G)$ such that
		\begin{itemize}
			\item For all $i\leqslant r$, $G[E_i]$ is co-connected and $|E_i|\geqslant 2$.
			\item For all $uv \in E(G)$, there exists $i$ such that $u,v \in E_i$ (we say that the hyperedge $E_i$ covers the edge $uv$).
		\end{itemize}
	\end{definition}

	\pbDef{co-Optimal Connecting HyperGraph (co-OCHG)}
	{A co-connected graph $G$.}
	{A co-Connecting Hypergraph that minimizes $\mathrm{cost}(H)=\sum_{E\in H}{(|E| - 2) }$.}	
	
	We prove the NP-hardness of this problem by reducing 3-SAT to it.
	We restrict ourselves to the version of 3-SAT where each variable has at least one positive and one negative occurrence and each clause has exactly 3 literals that are associated to different variables.   
	It is folklore that this restrictions of 3-SAT is NP-complete.
	
	Let $\mathscr F=\{c_1,\dots,c_m\}$ be an instance of 3-SAT with $n$ variables. 
	We will construct from $\mathscr F$ a graph $G_{\mathscr F}$ such that $\mathscr F$ is satisfiable if and only if $G_{\mathscr F}$ admits a co-covering hypergraph of cost $25m$.
	
	We start by describing how to construct $G_{\mathscr F}$. 
	To simplify the construction and the proofs, we give labels to some vertices and some edges. 
	The set of labels we use are $\{c_i,T_{i,x},F_{i,x} : i\leq m,\ x \text{ variable of } {\mathscr F} \}$.
	For each clause $c_i$ and each variable $x$ occurring in $c_i$, we create the gadget $g(x,c_i)$. 
	If $x$ occurs positively in $c_i$ then $g(x,c_i)$ is the graph depicted in Figure \ref{figxc}, otherwise, if $x$ occurs negatively in $c_i$ then $g(x,c_i)$ is the graph depicted in Figure \ref{fignxc}. 
	Each gadget $g(x,c_i)$ contains a vertex labelled $c_i$ and two edges labelled $T_{i,x}$ and $F_{i,x}$. 
	
\begin{figure}[h]
\begin{subfigure}{.49\textwidth}
\begin{center}
\begin{pspicture}(5,1.5)
\psline{-}(0.12,0.32)(4.6,0.32)
\psline{-}(1.4,0.32)(1.4,1.12)
\psline{-}(1.4,0.32)(2.36,1.12)
\psline{-}(3.32,0.32)(2.36,1.12)
\pscurve{-}(1.4,0.32)(2.36,0)(3.32,0.32)

\rput(-0.20,0.72){$F_{i,x}$}
\rput(4.9,0.72){$T_{i,x}$}
\rput(2.12,1.2){$c_i$}

\psline{-}(0.12,0.32)(0.12,1.12)
\psline{-}(4.6,0.32)(4.6,1.12)
\pscircle[linecolor=black, fillstyle=solid, fillcolor=white](4.6,1.12){0,12}
\pscircle[linecolor=black, fillstyle=solid, fillcolor=white](0.12,1.12){0,12}

\pscircle[linecolor=black, fillstyle=solid, fillcolor=white](0.12,0.32){0,12}
\pscircle[linecolor=black, fillstyle=solid, fillcolor=white](0.76,0.32){0,12}
\pscircle[linecolor=black, fillstyle=solid, fillcolor=white](1.4,0.32){0,12}
\pscircle[linecolor=black, fillstyle=solid, fillcolor=white](2.04,0.32){0,12}
\pscircle[linecolor=black, fillstyle=solid, fillcolor=white](2.68,0.32){0,12}
\pscircle[linecolor=black, fillstyle=solid, fillcolor=white](3.32,0.32){0,12}
\pscircle[linecolor=black, fillstyle=solid, fillcolor=white](3.96,0.32){0,12}
\pscircle[linecolor=black, fillstyle=solid, fillcolor=white](4.6,0.32){0,12}
\pscircle[linecolor=black, fillstyle=solid, fillcolor=white](1.4,1.12){0,12}
\pscircle[linecolor=black, fillstyle=solid, fillcolor=white](2.36,1.12){0,12}
\end{pspicture}
\caption{The gadget $g(x,c_i)$ if $x$ appears in $c_i$.}
\label{figxc}
\end{center}
\end{subfigure}
\begin{subfigure}{.49\textwidth}
\begin{center}
\begin{pspicture}(5,1.5)
\psline{-}(0.12,0.32)(4.6,0.32)
\psline{-}(3.32,0.32)(3.32,1.12)
\psline{-}(1.4,0.32)(2.36,1.12)
\psline{-}(3.32,0.32)(2.36,1.12)
\pscurve{-}(1.4,0.32)(2.36,0)(3.32,0.32)

\rput(-0.20,0.72){$F_{i,x}$}
\rput(4.9,0.72){$T_{i,x}$}
\rput(2.12,1.2){$c_i$}

\psline{-}(0.12,0.32)(0.12,1.12)
\psline{-}(4.6,0.32)(4.6,1.12)
\pscircle[linecolor=black, fillstyle=solid, fillcolor=white](4.6,1.12){0,12}
\pscircle[linecolor=black, fillstyle=solid, fillcolor=white](0.12,1.12){0,12}

\pscircle[linecolor=black, fillstyle=solid, fillcolor=white](0.12,0.32){0,12}
\pscircle[linecolor=black, fillstyle=solid, fillcolor=white](0.76,0.32){0,12}
\pscircle[linecolor=black, fillstyle=solid, fillcolor=white](1.4,0.32){0,12}
\pscircle[linecolor=black, fillstyle=solid, fillcolor=white](2.04,0.32){0,12}
\pscircle[linecolor=black, fillstyle=solid, fillcolor=white](2.68,0.32){0,12}
\pscircle[linecolor=black, fillstyle=solid, fillcolor=white](3.32,0.32){0,12}
\pscircle[linecolor=black, fillstyle=solid, fillcolor=white](3.96,0.32){0,12}
\pscircle[linecolor=black, fillstyle=solid, fillcolor=white](4.6,0.32){0,12}
\pscircle[linecolor=black, fillstyle=solid, fillcolor=white](3.32,1.12){0,12}
\pscircle[linecolor=black, fillstyle=solid, fillcolor=white](2.36,1.12){0,12}
\end{pspicture}
\caption{The gadget $g(x,c_i)$ if $\bar x$ appears in $c_i$.}
\label{fignxc}
\end{center}
\end{subfigure}
\caption{}
\end{figure}

We then create a new vertex for each clause $c_i$ that we connect to the three vertices labelled $c_i$ and to an additional vertex of degree 1. 
	We thus have for each clause a graph like the one depicted in Figure \ref{figclause} that we call $g(c_i)$.
	
\begin{figure}[h]
	
\begin{center}
\begin{pspicture}(5.84,6.5)
\psline{-}(2.92,2.52)(2.92,1.12)
\psline{-}(2.92,2.52)(1.536,3.52)
\psline{-}(2.92,2.52)(4.304,3.52)
\psline{-}(2.92,2.52)(4.304,1.92)

\psline{-}(0.12,0.32)(5.72,0.32)

\psline{-}(0.12,0.32)(0.12,1.12)
\psline{-}(5.72,0.32)(5.72,1.12)

\psline{-}(1.72,0.32)(1.72,1.12)
\psline{-}(1.72,0.32)(2.92,1.12)
\psline{-}(4.12,0.32)(2.92,1.12)
\pscurve{-}(1.72,0.32)(2.92,0)(4.12,0.32)


\pscircle[linecolor=black, fillstyle=solid, fillcolor=white](0.12,0.32){0,12}
\pscircle[linecolor=black, fillstyle=solid, fillcolor=white](0.92,0.32){0,12}
\pscircle[linecolor=black, fillstyle=solid, fillcolor=white](1.72,0.32){0,12}
\pscircle[linecolor=black, fillstyle=solid, fillcolor=white](2.52,0.32){0,12}
\pscircle[linecolor=black, fillstyle=solid, fillcolor=white](3.32,0.32){0,12}
\pscircle[linecolor=black, fillstyle=solid, fillcolor=white](4.12,0.32){0,12}
\pscircle[linecolor=black, fillstyle=solid, fillcolor=white](4.92,0.32){0,12}
\pscircle[linecolor=black, fillstyle=solid, fillcolor=white](5.72,0.32){0,12}

\pscircle[linecolor=black, fillstyle=solid, fillcolor=white](0.12,1.12){0,12}
\pscircle[linecolor=black, fillstyle=solid, fillcolor=white](5.72,1.12){0,12}
\pscircle[linecolor=black, fillstyle=solid, fillcolor=white](1.72,1.12){0,12}
\pscircle[linecolor=black, fillstyle=solid, fillcolor=white](2.92,1.12){0,12}

\psline{-}(-0.56,1.496)(2.24,6.344)
\psline{-}(0.24,2.88)(1.536,3.52)
\psline{-}(0.24,2.88)(0.936,2.48)
\psline{-}(1.44,4.96)(1.536,3.52)
\pscurve{-}(0.24,2.88)(0.568,4.08)(1.44,4.96)

\psline{-}(-0.56,1.496)(-1.256,1.896)
\psline{-}(2.24,6.344)(1.544,6.744)

\pscircle[linecolor=black, fillstyle=solid, fillcolor=white](-0.56,1.496){0,12}
\pscircle[linecolor=black, fillstyle=solid, fillcolor=white](-0.16,2.192){0,12}
\pscircle[linecolor=black, fillstyle=solid, fillcolor=white](0.24,2.88){0,12}
\pscircle[linecolor=black, fillstyle=solid, fillcolor=white](0.64,3.576){0,12}
\pscircle[linecolor=black, fillstyle=solid, fillcolor=white](1.04,4.264){0,12}
\pscircle[linecolor=black, fillstyle=solid, fillcolor=white](1.44,4.96){0,12}
\pscircle[linecolor=black, fillstyle=solid, fillcolor=white](1.84,5.656){0,12}
\pscircle[linecolor=black, fillstyle=solid, fillcolor=white](2.24,6.344){0,12}

\pscircle[linecolor=black, fillstyle=solid, fillcolor=white](1.536,3.52){0,12}
\pscircle[linecolor=black, fillstyle=solid, fillcolor=white](0.936,2.48){0,12}

\pscircle[linecolor=black, fillstyle=solid, fillcolor=white](-1.256,1.896){0,12}
\pscircle[linecolor=black, fillstyle=solid, fillcolor=white](1.544,6.744){0,12}

\psline{-}(6.4,1.496)(3.6,6.344)
\psline{-}(5.6,2.88)(4.304,3.52)
\psline{-}(4.4,4.96)(3.704,4.56)
\psline{-}(4.4,4.96)(4.304,3.52)
\pscurve{-}(5.6,2.88)(5.272,4.08)(4.4,4.96)

\psline{-}(6.4,1.496)(7.096,1.896)
\psline{-}(3.6,6.344)(4.296,6.744)

\pscircle[linecolor=black, fillstyle=solid, fillcolor=white](6.4,1.496){0,12}
\pscircle[linecolor=black, fillstyle=solid, fillcolor=white](6,2.192){0,12}
\pscircle[linecolor=black, fillstyle=solid, fillcolor=white](5.6,2.88){0,12}
\pscircle[linecolor=black, fillstyle=solid, fillcolor=white](5.2,3.576){0,12}
\pscircle[linecolor=black, fillstyle=solid, fillcolor=white](4.8,4.264){0,12}
\pscircle[linecolor=black, fillstyle=solid, fillcolor=white](4.4,4.96){0,12}
\pscircle[linecolor=black, fillstyle=solid, fillcolor=white](4,5.656){0,12}
\pscircle[linecolor=black, fillstyle=solid, fillcolor=white](3.6,6.344){0,12}

\pscircle[linecolor=black, fillstyle=solid, fillcolor=white](4.304,3.52){0,12}
\pscircle[linecolor=black, fillstyle=solid, fillcolor=white](3.704,4.56){0,12}

\pscircle[linecolor=black, fillstyle=solid, fillcolor=white](2.92,2.52){0,12}
\pscircle[linecolor=black, fillstyle=solid, fillcolor=white](4.304,1.92){0,12}

\pscircle[linecolor=black, fillstyle=solid, fillcolor=white](7.096,1.896){0,12}
\pscircle[linecolor=black, fillstyle=solid, fillcolor=white](4.296,6.744){0,12}

\rput(-0.18,0.76){$F_{i,x}$}
\rput(6.02,0.76){$T_{i,x}$}
\rput(2.68,1.2){$c_i$}

\rput(-1.1,1.426){$T_{i,z}$}
\rput(2.02,6.724){$F_{i,z}$}
\rput(1.746,3.68){$c_i$}

\rput(6.94,1.426){$F_{i,y}$}
\rput(3.77,6.729){$T_{i,y}$}
\rput(4.094,3.68){$c_i$}

\rput(4.52,0.8){\large{$g(x,c_i)$}}
\rput(5.52,5.48){\large{$g(y,c_i)$}}
\rput(0.4,5.48){\large{$g(z,c_i)$}}

\end{pspicture}
\caption{The clause-gadget associated to the clause $c_i=(x\lor \lnot y \lor \lnot z)$.}
\label{figclause}
\end{center}

\end{figure}

	Finally, for each variable $x$, we do the following.
	Let $c_{i_1},\dots,c_{i_\ell}$ be the clause where $x$ appears.
	Observe that $\ell\geqslant 2$ since every variable has a positive and a negative occurrence. 
	For each $j\leq \ell$, we merge the edge labelled $T_{i_j,x}$ in $g(x,c_{i_j})$ with the edge labelled $F_{i_{k},x}$ in $g(x,c_{i_{k}})$ (where $k=j+1\mod \ell$) such that the resulting edge has an extremity of degree one. 
	We consider that this edge has both $T_{i_j,x}$ and $F_{i_{k},x}$ as labels.
	For example, if a variable $x$ appears positively in the clauses $c_1$ and $c_4$ and negatively in the clause $c_3$, the Figure \ref{figvar} depicts what the graph looks like around the gadget associated to the variable $x$.

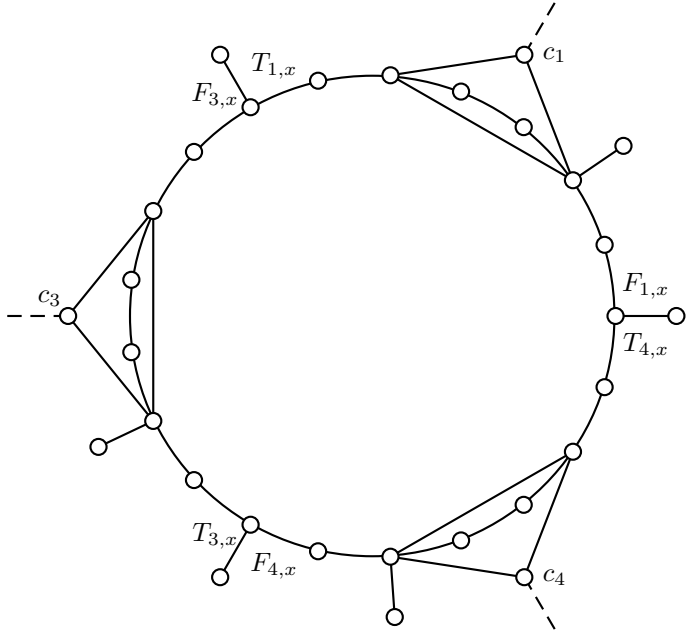
\begin{figure}[h]
\begin{center}
\begin{pspicture}(6.4,8)
\psline{-}(6.4,4)(7.2,4)
\psline{-}(5.84,5.8)(6.504,6.256)
\psline{-}(5.84,5.8)(5.2,7.464)
\psline{-}(3.44,7.192)(5.2,7.464)
\psline{-}(1.6,6.768)(1.2,7.464)
\psline{-}(0.32,5.392)(-0.8,4)
\psline{-}(0.32,2.608)(-0.8,4)
\psline{-}(0.32,2.608)(-0.4,2.264)
\psline{-}(1.6,1.232)(1.2,0.536)
\psline{-}(3.44,0.808)(3.496,0.008)
\psline{-}(3.44,0.808)(5.2,0.536)
\psline{-}(5.84,2.2)(5.2,0.536)

\psline{-}(5.84,5.8)(3.44,7.192)
\psline{-}(0.32,5.392)(0.32,2.608)
\psline{-}(3.44,0.808)(5.84,2.2)

\psline[linestyle=dashed]{-}(5.6,8.152)(5.2,7.464)
\psline[linestyle=dashed]{-}(-0.8,4)(-1.6,4)
\psline[linestyle=dashed]{-}(5.6,-0.152)(5.2,0.536)

\pscircle[linecolor=black](3.2,4){3.2}

\pscircle[linecolor=black, fillstyle=solid, fillcolor=white](6.4,4){0,12}
\pscircle[linecolor=black, fillstyle=solid, fillcolor=white](6.256,4.944){0,12}
\pscircle[linecolor=black, fillstyle=solid, fillcolor=white](5.84,5.8){0,12}
\pscircle[linecolor=black, fillstyle=solid, fillcolor=white](5.192,6.504){0,12}
\pscircle[linecolor=black, fillstyle=solid, fillcolor=white](4.368,6.976){0,12}
\pscircle[linecolor=black, fillstyle=solid, fillcolor=white](3.44,7.192){0,12}
\pscircle[linecolor=black, fillstyle=solid, fillcolor=white](2.488,7.12){0,12}
\pscircle[linecolor=black, fillstyle=solid, fillcolor=white](1.6,6.768){0,12}
\pscircle[linecolor=black, fillstyle=solid, fillcolor=white](0.856,6.176){0,12}
\pscircle[linecolor=black, fillstyle=solid, fillcolor=white](0.32,5.392){0,12}
\pscircle[linecolor=black, fillstyle=solid, fillcolor=white](0.032,4.48){0,12}
\pscircle[linecolor=black, fillstyle=solid, fillcolor=white](0.032,3.52){0,12}
\pscircle[linecolor=black, fillstyle=solid, fillcolor=white](0.32,2.608){0,12}
\pscircle[linecolor=black, fillstyle=solid, fillcolor=white](0.856,1.824){0,12}
\pscircle[linecolor=black, fillstyle=solid, fillcolor=white](1.6,1.232){0,12}
\pscircle[linecolor=black, fillstyle=solid, fillcolor=white](2.488,0.88){0,12}
\pscircle[linecolor=black, fillstyle=solid, fillcolor=white](3.44,0.808){0,12}
\pscircle[linecolor=black, fillstyle=solid, fillcolor=white](4.368,1.024){0,12}
\pscircle[linecolor=black, fillstyle=solid, fillcolor=white](5.192,1.496){0,12}
\pscircle[linecolor=black, fillstyle=solid, fillcolor=white](5.84,2.2){0,12}
\pscircle[linecolor=black, fillstyle=solid, fillcolor=white](6.256,3.056){0,12}

\pscircle[linecolor=black, fillstyle=solid, fillcolor=white](7.2,4){0,12}
\pscircle[linecolor=black, fillstyle=solid, fillcolor=white](5.2,7.464){0.12}
\pscircle[linecolor=black, fillstyle=solid, fillcolor=white](6.504,6.256){0,12}
\pscircle[linecolor=black, fillstyle=solid, fillcolor=white](1.2,7.464){0,12}
\pscircle[linecolor=black, fillstyle=solid, fillcolor=white](-0.8,4){0,12}
\pscircle[linecolor=black, fillstyle=solid, fillcolor=white](-0.4,2.264){0,12}
\pscircle[linecolor=black, fillstyle=solid, fillcolor=white](1.2,0.536){0,12}
\pscircle[linecolor=black, fillstyle=solid, fillcolor=white](3.496,0.008){0,12}
\pscircle[linecolor=black, fillstyle=solid, fillcolor=white](5.2,0.536){0,12}

\rput(6.8,4.4){$F_{1,x}$}
\rput(6.8,3.6){$T_{4,x}$}
\rput(1.912,7.32){$T_{1,x}$}
\rput(1.136,6.92){$F_{3,x}$}
\rput(1.912,0.68){$F_{4,x}$}
\rput(1.136,1.08){$T_{3,x}$}

\rput(5.6,7.464){$c_1$}
\rput(-1.04,4.24){$c_3$}
\rput(5.6,0.536){$c_4$}

\end{pspicture}
\caption{The gadgets associated to the variable $x$.}
\label{figvar}
\end{center}
\end{figure}

	By connecting all the gadgets $g(x,c_i)$ as described above, we obtain the gadget graph $G_{\mathscr F}$. 
	We may assume that $G_{\mathscr F}$ is connected. Otherwise, this means that ${\mathscr F}$ is the conjunction of two formulas that share no common variables and ${\mathscr F}$ is satisfiable if and only if those two formulas are.
	Observe that $G_{\mathscr F}$ is trivially co-connected. 
	Moreover, the size of $G_{\mathscr F}$ is polynomial in $n$ and $m$. 
	
	Now, we prove that $\mathscr F$ is satisfiable if and only if $G_{\mathscr F}$ admits a co-covering hypergraph of cost $25m$.
	We start with the following lemma which proves the existence of an optimal co-covering hypergraph where every hyperedge is contained in the vertex set of some clause-gadget.
	
	\begin{lemma} \label{lem:specialsolution}
		There exists an optimal co-connecting hypergraph $H$ of $G_{\mathscr F}$ such that $H= H_1\cup H_2 \cup \dots \cup H_m$  and for all $i\leq m$, we have $V(H_i)\subseteq V(g(c_i))$.
	\end{lemma}
	
	\begin{proof}
		
		In the graph $G_{\mathscr F}$, the intersection between two clause-gadgets only contains labelled edges. 
		Thus, if an hyperedge $E$ is not included in any clause-gadget, it means that $E$ covers at least two non-labelled edges from two distinct clause-gadgets.
		The Figure \ref{figjun} depicts what we call the junction between two clause-gadgets and names the vertices of interest in this proof. 
		Here, the clause-gadget $g(c_1)$ contains the vertices $v_0,v_1,v_2$ and $v_3$ and $g(c_2)$ contains $v_2,v_3,v_4$ and $v_5$. 
		
		\begin{figure}[h!]
				\begin{center}
					\begin{pspicture}(4.5,1.85)
					\psline{-}(0.15,0.4)(1.15,0.4)
					\psline{-}(1.15,0.4)(2.15,0.4)
					\psline{-}(2.15,0.4)(3.15,0.4)
					\psline{-}(3.15,0.4)(4.15,0.4)
					\psline{-}(2.15,0.4)(2.15,1.4)
					\rput(2.55,0.85){$F_{1,x}$}
					\rput(1.8,0.85){$T_{2,x}$}
					\pscircle[linecolor=black, fillstyle=solid, fillcolor=white](0.15,0.4){0,24}
					\pscircle[linecolor=black, fillstyle=solid, fillcolor=white](1.15,0.4){0,24}
					\pscircle[linecolor=black, fillstyle=solid, fillcolor=white](2.15,0.4){0,24}
					\pscircle[linecolor=black, fillstyle=solid, fillcolor=white](3.15,0.4){0,24}
					\pscircle[linecolor=black, fillstyle=solid, fillcolor=white](4.15,0.4){0,24}
					
					\pscircle[linecolor=black, fillstyle=solid, fillcolor=white](2.15,1.4){0,24}
					
					\rput(0.15,0.4){$v_0$}
					\rput(1.15,0.4){$v_1$}
					\rput(2.15,0.4){$v_2$}
					\rput(3.15,0.4){$v_4$}
					\rput(4.15,0.4){$v_5$}
					\rput(2.15,1.4){$v_3$}
					
					\rput(0.65,1.2){$c_1$}
					\rput(3.65,1.2){$c_2$}
	
					\end{pspicture}
					\caption{The junction between the clause-gadgets of $c_1$ and $c_2$. }
					\label{figjun}
				\end{center}
		\end{figure}
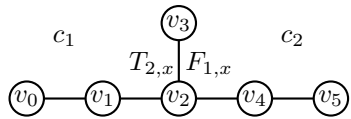

		Since $G_{\mathscr F}$ is connected and co-connected, we know by Lemma \ref{lem:co-connected} that it admits an optimal co-connecting hypergraph $H=\{E_1,\ldots,E_k\}$ such that for all $i$, $G_{\mathscr F}[E_i]$ is connected.
		The only way for an hyperedge $E_i$, such that $G_{\mathscr F}[E_i]$ is connected, to cover non-labelled edges in several clause-gadgets is to contain the vertices labelled $v_1,v_2$ and $v_4$ at the junction between two clause-gadgets. 
		In this case, we say that $E_i$ covers the junction between these two gadgets.
		Let us assume that an hyperedge $E_i$ covers the junction between two clause-gadgets $g(c_1)$ and $g(c_2)$. We make no assumption on whether $x$ appears positively or negatively in $c_1$ and $c_2$. 
		
		By definition, $E_i$ contains the vertices $v_1,v_2$ and $v_4$. 
		If $E_i$ does not contain $v_3$, there exists an hyperedge $E_j$ that covers the edge $\{v_2,v_3\}$ has to contain at least one of $v_1$ and $v_4$ (in order to induce a connected and co-connected subgraph of $G_{\mathscr F}$) and therefore shares at least two common vertices with $E_i$. 
		This means that we can merge $E_i$ and $E_j$ without increasing the cost of the solution and thus, we can assume that $E_i$ contains $v_3$. 
		However, $G_{\mathscr F}[\{v_1,v_2,v_3,v_4\}]$ is still not co-connected and $E_i$ has to contain other vertices. 
		By connectivity, $E_i$ must contain at least one of $v_0$ and $v_5$.
		\begin{itemize}
			\itemb If $E_i$ contains only one of $\{v_0,v_5\}$, say $v_0$, we remove $v_4$ from $E_i$ and add $v_2$ to the hyperedge $E_j$ that covers $\{v_4,v_5\}$. The hyperedges $E_i$ and $E_j$ still induce co-connected subgraphs of $G_{\mathscr F}$, cover the same edges as before and the cost of $H$ does not increase.
			\itemb If $G_{\mathscr F}[E_i\setminus\{v_2,v_3\}]$ is connected, we replace $E_i$ in the solution by the hyperedges $F_1=\{v_0,v_1,v_2,v_3\}$ and $F_2=E_i\setminus\{v_1,v_3\}$. The hyperedges $F_1$ and $F_2$ have the same cost as $E_i$ and cover the same edges, $G_{\mathscr F}[F_1]$ is co-connected and so is $G_{\mathscr F}[F_2]$ (all the vertices are adjacent to $v_2$ in $\bar{G_{\mathscr F}[F_2]}$ except $v_4$ that can be connected to $v_2$ through $v_0$)	. Let us also note that both $G_{\mathscr F}[F_1]$ and $G_{\mathscr F}[F_2]$ are connected.
			\itemb If $E_i$ contains both $v_0$ and $v_5$ but $G_{\mathscr F}[E_i\setminus\{v_2,v_3\}]$ is not connected, we know it has two connected components $C_1$ and $C_2$ (since removing a vertex set of degree $k$ cannot create more than $k$ connected components). We replace $E_i$ by $F_1=\{v_2,v_3\}\cup C_1$ and $F_2=\{v_2,v_3\}\cup C_2$. The hyperedges $F_1$ and $F_2$ have the same cost as $E_i$ and cover the same edges. Furthermore, both $\bar{G_{\mathscr F}[F_1]}$ and $\bar{G_{\mathscr F}[F_2]}$ are connected since every vertex is adjacent to $v_3$ except $v_2$ that can be connected to $v_3$ through $v_0$ and $v_5$ in $C_1$ and $C_2$. We also notice that $G_{\mathscr F}[F_1]$ and $G_{\mathscr F}[F_2]$ are both connected.
		\end{itemize}
		In any case, we can build an optimal co-connecting hypergraph where $G_{\mathscr F}[E_i]$ is still connected for all $i$ and the hyperedges cover strictly fewer junctions. We can iterate this process until $H$ satisfies the lemma.
	\qed
	\end{proof}

	Let $H=H_1\cup \dots \cup H_m$ be an optimal co-connecting hypergraph of $G$ such that for all $i\leq m$, we have $V(H_i)\subseteq V(g(c_i))$.
	
	Observe that the labelled edges are the only edges of $G_{\mathscr F}$ to belong to several clause-gadgets.
	Thus, for each $i\leq m$, the non-labelled edges of $g(c_i)$ must be covered by $H_i$.
	Consequently, the cost of $H_i$ is fully determined by which labelled edges of $g(c_i)$ it covers.
	We want to prove that $G_{\mathscr F}$ is satisfiable if and only if the labelled edges can be covered in a way such that each $H_j$ has cost $25$.
	
	Let $c_i$ be a clause of $\mathscr F$ and let us study the cost of $H_i$ in function of which labelled edges it covers. 
	Let $x$ be a variable of $c_i$. 
	We recall that the gadget $g(x,c_i)$ differs depending on whether $x$ appears positively or negatively in $c_i$ but in both cases, the gadget has an edge labelled $F_{i,x}$ and one labelled $T_{i,x}$.
The subgraph of $g(x,c_i)$ induced by $H_i$ can take four values (up to isomorphims) depending on which of the following situations occurs:
	\begin{itemize}
		\item $H_i$ covers neither $T_{i,x}$ nor $F_{i,x}$. We call this configuration $N$ (for “none”).
		\item $H_i$ covers both $T_{i,x}$ and $F_{i,x}$. We call this configuration $B$ (for “both”).
		\item $H_i$ covers $T_{i,x}$ and $x$ appears positively in $c_i$ or $H_i$ covers $F_{i,x}$ and $x$ appears negatively in $c_i$.
		We call this configuration $S$ (for “satisfied”).
		\item $H_i$ covers $T_{i,x}$ and $x$ appears negatively in $c_i$ or $H_i$ covers $F_{i,x}$ and $x$ appears positively in $c_i$.
		We call this configuration $U$ (for “unsatisfied”).
	\end{itemize}
	Hence, the edges that $H_i$ covers are determined (up to isomorphism) by the configurations encountered for each of the three variables that appear in $c_i$. 
	Since the clause-gadget is symmetric, the order does not matter: the configuration $SUN$ is exactly the same as the configuration $NSU$. Thus, we find that $H_i$ can cover 20 different sets of edges up to isomorphisms. 
	We determined the optimal values of $\cost(H_i)$ for each case via a computer-assisted exhaustive search. The results are the following:
	
	\begin{center}
		\begin{tabular}{|c|c| |c|c| |c|c| |c|c|}
			\hline
			Configuration & Minimal cost & conf.& min. &conf. & min. &conf. & min.\\
			\hline
			$BBB$ & 28 & $BUS$ & 26 & $UUU$ & 26 & $UNN$ & 25 \\  
			\hline
			$BBU$ & 27 & $BUN$ & 26 & $UUS$ & 25 & $SSS$ & 25 \\
			\hline
			$BBS$ & 27 & $BSS$ & 26 & $UUN$ & 25 & $SSN$ & 25 \\
			\hline
			$BBN$ & 27 & $BSN$ & 26 & $USS$ & 25 & $SNN$ & 25 \\
			\hline
			$BUU$ & 26 & $BNN$ & 26 & $USN$ & 25 & $NNN$ & 25 \\
			\hline
		\end{tabular}
	\end{center}
	
	The first observation we make is that the optimal value of $\cost(H_i)$ is necessarily at least 25 and an optimal co-connecting hypergraph on $G_{\mathscr F}$ therefore always costs at least $25m$. We now investigate the case where the optimal cost is exactly $25m$.
	To this end, we suppose that $H$ has a cost of $25m$.
	
	We note that every configuration that contains a $B$ costs at least 26. 
	Thus, we know that for each $H_i$ and each $x$ appearing in $c_i$, $H_i$ covers at most one of the two labelled edges of $g(x,c_i)$.
	
	Let us now look at the gadgets associated to a variable $x$ that appears in $\ell$ clauses (cf. Figure \ref{figvar}). 
	For all $j$ such that $x$ appears in $c_j$, the hypergraph $H_j$ either covers the two labelled edges of $g(x,c_j)$ ($B$), one ($S$ or $U$) or none ($N$). 
	Since every edge must be covered at least once, this means that a solution where no configuration involves $B$ also does not feature a configuration involving $N$. 
	Hence, for every $H_i$, the only configurations that occur are $S$ and $U$. 

	Let us suppose that $H_i$ covers the edge $T_{i,x}$. 
	Since the configuration $B$ is impossible, we know that $H_i$ does not cover the edge $F_{i,x}$. 
	Let $T_{j,x}$ the other label of the edge $F_{i,x}$. 
	Since this edge has to be covered, this means that $H_j$ must cover the edge $T_{j,x}$ and because the configuration $B$ is impossible, it cannot cover the edge $F_{j,x}$.
	 We can prove by induction that for each variable $x$ either, for all gadget $g(x,c_i)$, $H_i$ covers the edge $T_{i,x}$ or for all gadget $g(x,c_i)$, $H_i$ covers the edge $F_{i,x}$.
	 In the first case, we say that the variable $x$ is set to $\True$, and in the second case, to $\False$. 
	 If the variable $x$ is set to $\True$, this means all its positive occurrence will lead to a $S$ configuration in the clause where it appears and conversely.
	
	Finally, we notice that the cost of an optimal co-connecting hypergraph on the configurations $SSS$, $SSU$ and $SUU$ is 25 while it is 26 on the configuration $UUU$. Therefore, there exists a solution of cost $25m$ if and only of there exists a way to affect all the variables to either $\True$ or $\False$ such that every clause is satisfied by at least one variable, which comes down to saying that the formula $\mathscr{F}$ is satisfiable. 
	
	This proves that \textbf{co-OCGH} and therefore \textbf{OCGH} and \textbf{MCTS} are all NP-hard.
	Moreover, Lichtenstein proved in \cite{planar} that 3-SAT remains NP-complete when restricted to formulas whose incidence graph is planar. 
	The incidence graph of a formula $\mathscr F$ is the bipartite graph representing the relation of belonging between the variables and the clauses of $\mathscr F$.
	Clearly, if the incidence graph of $\mathscr F$ is planar then $G_{\mathscr F}$ is planar too.
	We conclude that \textbf{MCTS} is NP-hard even on co-planar graphs.
	 
\section*{Conclusion}

Our work proves that finding a minimum connecting transition set is NP-hard even on co-planar graphs. 
This notably implies the NP-hardness of other problems that generalizes this one such as finding a minimum connecting transition set in a graph that already has forbidden transitions. 

A lot of our results suggest that the density of the graph has an impact on the complexity of \textbf{MCTS}.
Consequently, it would be interesting to study the complexity of this problem on sparse graphs such as planar graphs or graphs with bounded treewith.

Further works could lead us to generalize this study to directed graphs, that are more suitable for many practical applications. 
Another interesting continuation of this work would also be the study of low-stretch connecting transition sets, a problem that is already well-studied for minimal spanning trees \cite{stretch}. 
Intuitively, it consists in looking for a subset of transitions $T$ such that the shortest $T$-compatible path between two vertices is not much longer than the shortest path in the graph with no forbidden transitions, which is also an important criteria of robustness.

\section*{Acknowledgments} 

The authors would like to thank Marthe Bonamy, Mamadou M. Kanté, Arnaud Pêcher, Théo Pierron and Xuding Zhu for the interest they showed for our work and for inspiring discussions.

\bibliography{biblio}

\begin{thebibliography}{10}

\bibitem{Ahmed}
Mustaq Ahmed and Anna Lubiw.
\newblock Shortest paths avoiding forbidden subpaths.
\newblock In {\em 26th International Symposium on Theoretical Aspects of
  Computer Science, {STACS} 2009, February 26-28, 2009, Freiburg, Germany,
  Proceedings}, pages 63--74, 2009.

\bibitem{SCTM}
Thomas Bellitto.
\newblock Separating codes and traffic monitoring.
\newblock {\em Theoretical Computer Science}, 2017.

\bibitem{Daykin}
C.~C. Chen and David~E. Daykin.
\newblock Graphs with hamiltonian cycles having adjacent lines different
  colors.
\newblock {\em J. Comb. Theory, Ser. {B}}, 21(2):135--139, 1976.

\bibitem{Gutin}
Gregory Gutin and Eun~Jung Kim.
\newblock Properly coloured cycles and paths: Results and open problems.
\newblock In {\em Graph Theory, Computational Intelligence and Thought, Essays
  Dedicated to Martin Charles Golumbic on the Occasion of His 60th Birthday},
  pages 200--208, 2009.

\bibitem{Kante-algopath}
Mamadou~Moustapha Kant{\'{e}}, Christian Laforest, and Benjamin Mom{\`{e}}ge.
\newblock An exact algorithm to check the existence of (elementary) paths and a
  generalisation of the cut problem in graphs with forbidden transitions.
\newblock In {\em {SOFSEM} 2013: Theory and Practice of Computer Science, 39th
  International Conference on Current Trends in Theory and Practice of Computer
  Science, {\v{S}}pindler{\r{u}}v Ml{\'{y}}n, Czech Republic, January 26-31,
  2013. Proceedings}, pages 257--267, 2013.

\bibitem{Kante-grid}
Mamadou~Moustapha Kant{\'{e}}, Fatima~Zahra Moataz, Benjamin Mom{\`{e}}ge, and
  Nicolas Nisse.
\newblock Finding paths in grids with forbidden transitions.
\newblock In {\em Graph-Theoretic Concepts in Computer Science - 41st
  International Workshop, {WG} 2015, Garching, Germany, June 17-19, 2015,
  Revised Papers}, pages 154--168, 2015.

\bibitem{Kotzig}
Anton Kotzig.
\newblock Moves without forbidden transitions in a graph.
\newblock {\em Matematický časopis}, 18(1):76--80, 1968.

\bibitem{planar}
David Lichtenstein.
\newblock Planar formulae and their uses.
\newblock {\em {SIAM} J. Comput.}, 11(2):329--343, 1982.

\bibitem{stretch}
David Peleg.
\newblock Low stretch spanning trees.
\newblock In {\em Mathematical Foundations of Computer Science 2002, 27th
  International Symposium, {MFCS} 2002, Warsaw, Poland, August 26-30, 2002,
  Proceedings}, pages 68--80, 2002.

\bibitem{Sudakov}
Benny Sudakov.
\newblock Robustness of graph properties.
\newblock {\em arXiv}, 2016.

\bibitem{szeider}
Stefan Szeider.
\newblock Finding paths in graphs avoiding forbidden transitions.
\newblock {\em Discrete Applied Mathematics}, 126(2-3):261--273, 2003.

\end{thebibliography}
\bibliographystyle{plain}
	 
\end{document}